\documentclass[aps,pra,showpacs,twoside,twocolumn,10pt]{revtex4-2}
\usepackage[colorlinks=true, citecolor=blue, urlcolor=blue ]{hyperref}
\usepackage{times,epsfig,amssymb,amsfonts,amsmath,bm,subfigure,mathtools,amsthm,braket,soul,enumitem,xcolor,physics,graphics,graphicx,tikz}
\UseRawInputEncoding
\usepackage[normalem]{ulem}
\usepackage{comment}
\usepackage{epstopdf}
\usepackage{relsize}
\usepackage{orcidlink}
\newtheorem{theorem}{Theorem}
\newtheorem{corollary}[theorem]{Corollary}
\newtheorem{definition}{Definition}
\newtheorem{proposition}[theorem]{Proposition}

\newtheorem{lemma}{Lemma}

\begin{document}

\title{
%Quantum State Texture: Connections to Convex Resource Theories and Applications in Detecting Quantum Phase Transition\\
%Texture-based monotone: Probing resource theories and quantum phase transition\\
%Texture-based monotones: From Characterizing Resources to Detecting Quantum Phase Transition\\
%The Emergent Utility of Quantum State Texture: From Characterizing Resources to Detecting Quantum Phase Transition\\
%Utility of quantum state texture: Probing resource theories and quantum phase transition\\
{Role of quantum state texture in probing resource theories and quantum phase transitions}
}

\author{Ayan Patra, Tanoy Kanti Konar, Pritam Halder, Aditi Sen(De)}

\affiliation{ Harish-Chandra Research Institute,  A CI of Homi Bhabha National Institute, Chhatnag Road, Jhunsi, Prayagraj - $211019$, India}

\begin{abstract}

Building on the recently developed quantum state texture resource theory, we exhibit that the difference between maximum and minimum textures is a valid purity monotone in any dimension and provide a lower bound for existing purity measures. We introduce a texture-based resource monotone applicable across general convex resource theories, encompassing quantum coherence, non-stabilizerness, and entanglement. In particular, we propose the notion of non-local texture, which corresponds to the geometric measure of bipartite and multipartite entanglement in pure states. Furthermore, we demonstrate that the texture of the entire ground state or its subsystems can effectively signal quantum phase transitions in the Ising chain under both transverse and longitudinal magnetic fields, offering a powerful tool for characterizing quantum criticality.

%The theory of quantum state texture has recently emerged as a novel resource framework, with applications in areas such as quantum thermodynamics and gate identification in quantum circuits.  In this work, we explore the relationship between the quantum state texture and other resource theories, such as purity, as well as with the broader framework of convex resource theories (CRTs). We propose a texture-based measure of purity that satisfies all the essential properties of a valid purity monotone. Furthermore, we construct a texture-based resource monotone applicable to general CRTs, including coherence, non-stabilizerness, and entanglement resource theories. Moreover, to characterize non-local correlations in bipartite systems, we introduce the concept of non-local texture, which aligns with the geometric measure of bipartite entanglement. We further extend this notion to the multipartite scenarios. Finally, we show that texture is capable of detecting quantum phase transitions in many-body systems. As a concrete example, we analyze the Ising spin chain models and illustrate the utility of texture in identifying critical points.

\end{abstract}

\maketitle
\section{Introduction}
\label{sec:intro}

%In recent years, the resource-theoretic approach has gained significant traction in the quantum information community~\cite{chitambar_rmp_2019}. These frameworks are particularly appealing from a practical standpoint, as they typically focus on quantum operations that align with existing experimental limitations. 

Resource theories offer a structured framework for quantifying and comparing the resource content of quantum states under constraints on allowable operations~\cite{chitambar_rmp_2019}. This approach is particularly appealing in practice, as it aligns with experimentally feasible quantum operations.
For instance, different sets of freely accessible quantum operations, such as local operations and classical communication (LOCC), stabilizer operations (SO), thermal operations (TO), and incoherent operations (IO), naturally lead to distinct resource theories, namely those of entanglement~\cite{Bennett_PRA_1996,formation96,Vedral_PRL_1997,Vedral_PRA_1998,vidal_prl_1999,horodecki_rmp_2009}, non-stabilizerness~\cite{Veitch2014Jan,howard_prl_2017,knill2004,liu_prxq_2022,leone_pra_2024,bravyi_pra_2005}, athermality~\cite{brandao_prl_2013, Janzing_IJTP_2000,Masanes_NatureCom_2017, Brandao_PNAS_2015,Gour_NatureCom_2018,Ng_SIP_2018}, and coherence~\cite{baumgratz_prl_2014, streltsov_rmp_2017,winter_prl_2016,streltsov_prl_2015,Streltsov2018May}, respectively. Given that processes such as decoherence or certain quantum operations quickly diminish the quantum characteristics of a system~\cite{Horodecki_ROMP_2003,Yu_Science_2009,Liu_PRL_2017,Pal_JOPA_2015,Muhuri_arXiv_2023,Patra_arXiv_2024}, adopting a resource-theoretic perspective is a natural choice for analyzing quantum systems.
Interestingly, many of these theories are interconnected; for instance, incoherent operations cannot generate entanglement from incoherent states, linking the resource theories of coherence and entanglement ~\cite{streltsov_prl_2015}. Exploring such interconnections offers deep insights and exciting research opportunities in quantum science.
%Resource theories provide a structured framework for quantifying and comparing the resource content of quantum states, given constraints on the permissible quantum operations. Interestingly, these resource theories are often interconnected. 
%
%For example, in coherence theory, 
%it has been shown that entanglement cannot be created between two incoherent states through incoherent operations~\cite{streltsov_prl_2015}, thus linking the resource theories of coherence and entanglement. Therefore, exploring the connections between a new resource theory and existing ones is a compelling and meaningful direction of research.

Recently, Parisio~\cite{Parisio2024} introduced the resource theory of quantum state texture, which characterizes the quantum states using the real parts of their elements of the density matrix  in a chosen basis. Within this framework, a state is considered texture-less and thus free if all real parts of its density matrix elements are equal. The associated free operations are those that preserve the texture-less states. This theory also shows promise in practical applications, such as quantum gate identification in quantum circuits. Note further that like coherence, the texture resource theory is inherently basis-dependent. However, its connection to established resource theories, such as coherence, entanglement, or non-stabilizerness, remains unexplored. Establishing this connection and exploring its applications in detecting cooperative phenomena is crucial, as it may offer new perspectives and deeper insights into the structure and interplay of known resource theories through the lens of quantum state texture.

%Recently, Parisio~\cite{Parisio2024} proposed a novel resource theory called the resource theory of quantum state texture, which characterizes quantum states based on the real parts of their density matrix elements when represented in a chosen basis. In this work, the author developed the theoretical foundation of texture as a resource, identified the corresponding free operations, and demonstrated its potential application in quantum gate identification within quantum circuits. Specifically, in this resource theoretic framework, a state is deemed texture-less, and hence considered a free state, if all real parts of its density matrix elements are equal. Correspondingly, the free operations are those that preserve the texture-less state. Evidently, similar to coherence, the resource theory of quantum state texture is inherently basis dependent. However, the relationship between this texture resource theory and other well-established resource theories has not yet been explored -- a task that remains crucial. Uncovering such connections could pave the way for a deeper understanding of existing resource theories through the lens of quantum texture.

In this paper, we explore the connection between the quantum texture of a state and various resource theories (see Fig.~\ref{fig:schematics}), particularly focusing on the state’s purity and its resource content within a convex resource-theoretic framework. Building on this foundation, we introduce a texture-based measure of quantum purity (see Ref.~\cite{Wang_PRA_2025} for details on different texture measures) that satisfies the standard conditions of a purity monotone ~\cite{Streltsov2018May,Gour_PhysicsReports_2015}. Notably, this measure also provides a lower bound on the widely used \textit{R\'enyi}-$2$ purity~\cite{Streltsov2018May,Gour_PhysicsReports_2015}.

We further establish a general link between quantum state texture and a broad class of convex resource theories (CRTs)~\cite{chitambar_rmp_2019}, which are characterized by the property that all free pure states are interconvertible via free unitary operations. Prominent examples of such CRTs include entanglement, coherence, and non-stabilizerness.
To quantify resource content in this setting, we define a {\it texture-based resource monotone} by minimizing the texture of a quantum state over a set of basis choices relevant to the resource theory in question. We derive explicit and computable forms for the texture-based coherence and non-stabilizerness monotones. Furthermore, to characterize non-local correlations in multipartite systems, we introduce the notion of {\it non-local texture}, defined as the minimum texture across all local bases. We demonstrate that this quantity serves as a valid entanglement measure for both bipartite and multipartite pure states.

Identifying cooperative phenomena in quantum many-body systems remains a significant challenge, particularly when aiming for experimentally accessible methods. While quantum features such as entanglement \cite{horodecki_rmp_2009}, coherence \cite{streltsov_rmp_2017}, and quantum discord \cite{Modi2012,Bera_2017} have proven to be effective in detecting such phenomena, they often require complete knowledge of the quantum state, which is difficult to obtain in practical settings. In contrast, the quantum texture monotone offers a key advantage: it can be estimated without full-state tomography, making it a promising tool for probing many-body systems. To investigate this potential, we analyze the behavior of quantum texture in both integrable and non-integrable versions of the Ising chain \cite{sachdev_2011}. Specifically, we compute the texture of the ground state and its subsystems for the transverse-field Ising model and for the Ising model under an additional longitudinal magnetic field. For the integrable case, we derive an analytical expression for the ground state texture and show that it exhibits a noticeable change in curvature near the quantum phase transition point. For the non-integrable model, the addition of a longitudinal field breaks the system’s symmetry, and this symmetry-breaking is also reflected in the ground state texture. These findings demonstrate that quantum texture serves as an effective and experimentally friendly indicator of quantum phase transitions.

This paper is structured as follows. Sec.~\ref{sec:res_th_texture} provides a brief overview of the resource theory of quantum state texture. In Sec.~\ref{sec:texture_purity}, we introduce a texture-based measure of purity. Sec.~\ref{sec:texture_resource} presents the construction of texture-based resource measures for specific convex resource theories, with detailed examples illustrating non-stabilizerness, coherence, and entanglement. In Sec.~\ref{sec:texture_QPT}, we demonstrate how quantum state texture can serve as a tool for detecting quantum phase transitions in both integrable and non-integrable models. Finally, Sec.~\ref{sec:conclusion} presents our concluding remarks. 

% We also demonstrate that the concept of nonstabilizerness more evocatively referred to as the 'magic' of a quantum state can be quantified using the texture measure. Specifically, we show that minimizing the texture of a pure state over all Clifford unitary transformations serves as a valid measure of nonstabilizerness, effectively capturing its distance from the set of stabilizer states. Furthermore, we provide an exact formulation of this optimization for qubit systems.

% In addition to capturing the intrinsic properties of a quantum system through texture, we show that correlations between subsystems can also be quantified using this measure. In particular, we demonstrate that the entanglement of a quantum state can be detected via texture. To achieve this, we introduce a notion of non-local texture—a component of texture that cannot be eliminated through local operations. We prove that non-local texture is both a necessary and sufficient condition for entanglement, and that it can serve as a quantifier of bipartite entanglement. Extending this framework to multipartite systems, we further show that our texture-based approach captures genuine multipartite entanglement. In this context, the non-local texture naturally leads to a quantifier of genuine multipartite entanglement, aligning with established measures such as the Generalized Geometric Measure (GGM).

\section{Resource theory of quantum texture}
\label{sec:res_th_texture}
A resource theory is defined by imposing a specific restriction on the set of quantum operations that can be implementable without incurring any cost. Based on this restriction, for a given resource theory, one identifies the set of free states, which are those states that can be generated using only the free operations. Crucially, free operations must transform any free state into another free state. An essential component of any resource theory is the construction of a resource monotone, a function that can characterize the resource manipulation under the application of free operations. 

{\it Quantum state texture \cite{Parisio2024}}. 
%We begin by briefly reviewing the recently proposed resource-theoretic framework known as the texture of quantum states \cite{Parisio2024}. 
Let us consider a given density operator $\rho$ acting on $\mathbb C^d$ is written in a fixed basis, $\mathsf B\equiv\{\ket{i}\}_{i=1}^d$ as $\rho=\sum_{i,j}\rho_{ij}\ketbra{i}{j}$, where $d$ is the dimension of the Hilbert space.  The landscape of real part of $\rho$, $\Re(\rho)$, in the aforementioned basis can be observed by visualizing $\bra i\Re(\rho)\ket j=\Re(\rho_{ij})$ against $\{i,j\}$ 
%as shown in Fig.~\ref{fig:schematics} for a system of dimension $d=4$
(for a $3$D-graphical illustration, see Fig.~$1$ of Ref.~\cite{Parisio2024}). The texture of the landscape turns out to be useful in characterizing unknown quantum gates, especially CNOT gates, in universal circuit layers. 

{\it Free states and free operations in texture theory.} The only resourceless state in the resource theory of texture is the density matrix $\rho$ whose elements are all identical, i.e., $\rho_{ij}=1/d \,~ \forall ~\{i,j\}$. This corresponds to the pure state $\rho=s_1\equiv\ketbra{s_1}{s_1}$, where $\ket{s_1}=\frac{1}{\sqrt{d}}\sum_{i=1}^d\ket{i}$. Note that this state possesses maximal coherence \cite{streltsov_rmp_2017} with respect to the basis $\mathsf B$. A quantum operation $\Lambda_{\mathcal{T}}(*) = \sum_i K_i (*) K_i^\dagger$, satisfying $\sum_i K_i^\dagger K_i = I$, is considered as a free operation if it leaves $s_1$ invariant, i.e., $\Lambda_{\mathcal{T}}(s_1) = s_1$. Since $s_1$ is a pure state, this condition implies that each Kraus operator acts on $\ket{s_1}$ as $K_i \ket{s_1} = c_i \ket{s_1}~\forall~i$, where $c_i$ are constants.

% \begin{figure}
%     \centering
%     \includegraphics[width=\linewidth]{schem.pdf}
%     \caption{\textbf{Schematics for distinguishing a textured state from a texture-less one.} Real parts of all elements of a four-dimensional density matrix \(\left(\text{Re}[\rho_{ij}]\right)\) are displayed. In the case of a textured state (see left panel), the elements exhibit varying colors, indicating underlying unevenness in the landscape of the density matrix.
%     Conversely, a smooth landscape of a uniform color across all elements (see right panel), reflecting a texture-less state. In the computational basis, each element of the texture-less state's density matrix takes the value \(1/d\), where \(d\) is the dimension of the system (here, \(d=4\)).}
%     \label{fig:schematics}
% \end{figure}

\begin{figure}
    \centering
    \includegraphics[width=\linewidth]{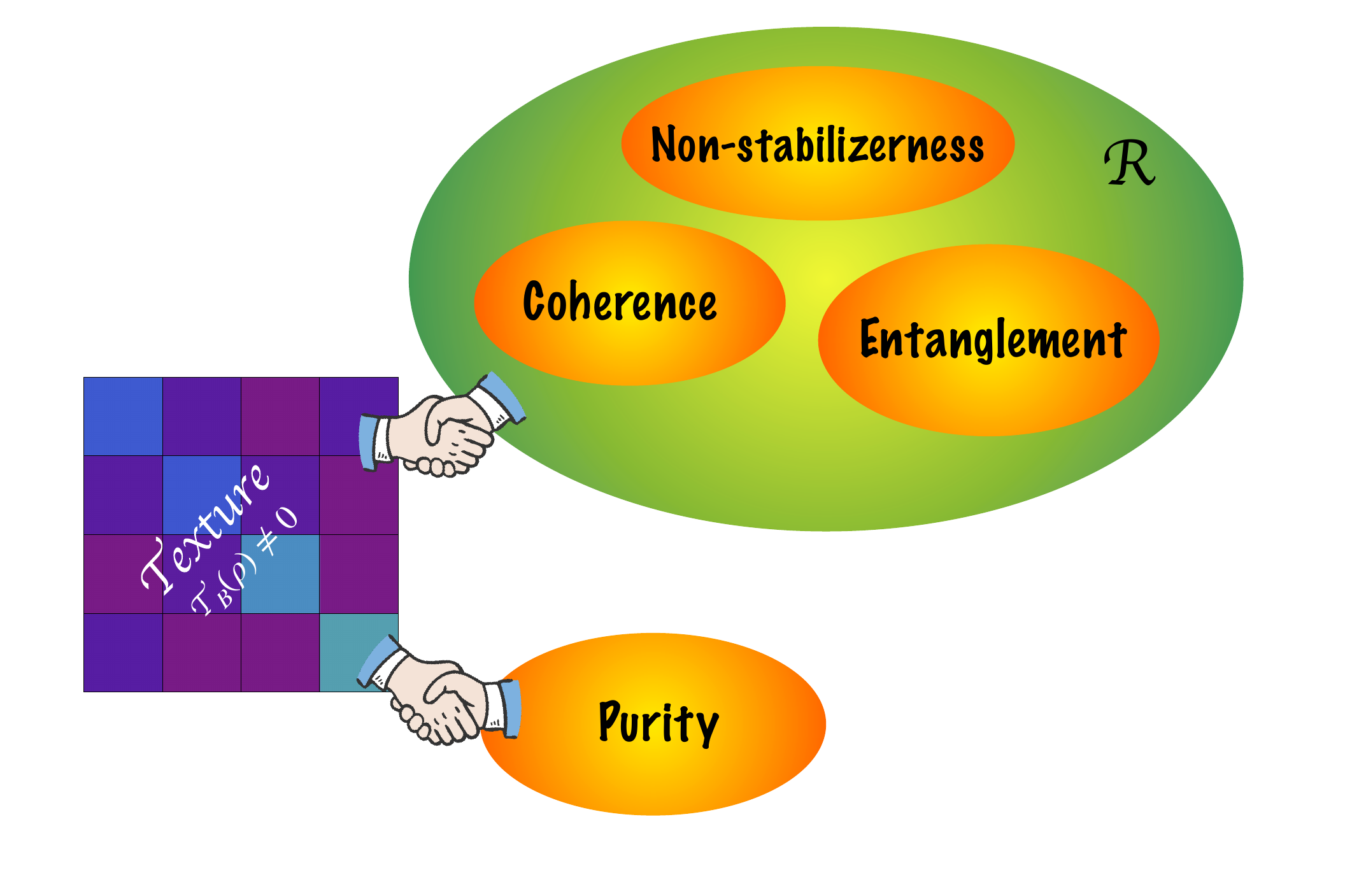}
    \caption{\textbf{Connecting quantum state texture to other resource theories.} Real parts of all elements of a four-dimensional density matrix \(\left(\text{Re}[\rho_{ij}]\right)\) are displayed. In the case of a textured state, the elements exhibit varying colors, indicating underlying structure in the density matrix. Conversely, the texture-less state would have a uniform color across all elements, reflecting a lack of structure. In a given basis, each element of the texture-less state's density matrix takes the value \(1/d\), where \(d\) is the dimension of the system (here, \(d=4\)). We connect texture resource theory to purity and other resource theories $\mathcal R$, which consists of free pure states that are interconnected by free unitaries of $\mathcal R$, e.g., non-stabilizerness, coherence, and entanglement.}
    \label{fig:schematics}
\end{figure}

{\it Texture monotone.} To serve as a valid measure of texture of a given state $\rho$, the texture monotone $\mathcal T$ must satisfy the following properties --- $(1)$ \textit{faithfulness}: $\mathcal{T}(\rho)\geq0$, where the equality sign holds if and only if $\rho=s_1$; $(2)$ \textit{monotonicity}: texture must not increase under free operation, i.e., $\mathcal T(\Lambda_{\mathcal T} (\rho))\leq \mathcal{T}(\rho)$ $\forall \rho$; and $(3)$ \textit{convexity}: $\mathcal T(\rho)$ should be a convex function of density matrices which means $\mathcal{T}(\sum_kp_k\rho_k)\leq\sum_kp_k\mathcal T(\rho_k)$. Note that while positivity and monotonicity are essential requirements for a valid resource monotone,  convexity is often desirable but not strictly necessary properties.

\noindent Let us define the measure of texture corresponding to the basis $\mathsf B$ as 
\begin{equation}
    \mathcal T_{\mathsf B}(\rho)= 1-\frac{1}{d}\sum_{i,j=1}^d\rho_{ij} = 1-\frac{1}{d}\mathcal E(\rho),
    \label{eq:texture_measure}
\end{equation}
where $\mathcal E(\rho)=\sum_{i,j=1}^d\rho_{ij}=d\bra{s_1}\rho\ket{s_1}$ is the sum of all the elements of $\rho$, known as the `grand sum'. Note that $\mathcal{E}(\rho)/d$ denotes the probability of projecting the state $\ket{s_1}$ onto the density matrix $\rho$, which implies that $\mathcal{E}(\rho)$ lies within the range $0 \leq \mathcal{E}(\rho) \leq d$, and consequently, $0\leq \mathcal T_{\mathsf B}\leq 1$. Note that \(\mathcal{T}_{\mathsf{B}}(\rho)\) can be expressed as \(\mathcal{T}_{\mathsf{B}}(\rho) = 1 - e^{-\mathfrak{R}(\rho)}\), where \(\mathfrak{R}(\rho) = -\ln\left(\mathcal{E}(\rho)/d\right)\) denotes the \textit{state rugosity}~\cite{Parisio2024}, and takes values in the range \(0 \leq \mathfrak{R} \leq \infty\). The \textit{faithfulness} and \textit{monotonicity} of \(\mathcal{T}_{\mathsf{B}}(\rho)\) directly follow from the fact that \(\mathfrak{R}(\rho)\) itself is a faithful texture monotone under general free operations \(\Lambda_{\mathcal{T}}\)~\cite{Parisio2024}. Furthermore, since $\mathcal E(*)$ is a linear functional, it immediately follows that \(\mathcal T_\mathsf{B}(\sum_kp_k\rho_k)=\sum_k p_k \mathcal{T}_\mathsf{B}(\rho_k)\), implying that \(\mathcal T_\mathsf{B}(\sum_kp_k\rho_k)\) is linear in its argument.

\begin{comment}
Next, we show  that, $\mathcal M=\mathcal T_{\mathsf B}$ is, in fact, a valid texture measure that fulfills the aforementioned criteria --
\begin{itemize}
    \item $\mathcal{T}_{\mathsf{B}}(\rho)=0$ iff $\rho=s_1$ which means $\mathcal{T}_{\mathsf{B}}$ is faithful.

    \item The texture measure $\mathcal{T}_{\mathsf{B}}$ can be expressed as $\mathcal{T}_{\mathsf{B}}(\rho) = 1 - e^{-\mathfrak{R}(\rho)}$, where $\mathfrak{R}(\rho)=-\ln\left(\mathcal E(\rho)/d\right)$ represents the state rugosity which is nonincreasing under general free operations \cite{Parisio2024}. Since $\mathfrak{R}(\rho)$ has a one-to-one correspondence with $\mathcal{T}_{\mathsf{B}}(\rho)$, it satisfies condition $(2)$.

    \item $\mathcal E(*)$ being a linear functional, we can write 
    \begin{eqnarray}
        \nonumber\mathcal T_\mathsf{B}(\sum_kp_k\rho_k)&=&\sum_kp_k-\sum_kp_k\mathcal{E}(\rho_k)/d,\\\nonumber &=&\sum_k p_k \mathcal{T}_\mathsf{B}(\rho_k).
    \end{eqnarray}
     Therefore, $\mathcal T_\mathsf{B}$ saturates condition $(3)$.
\end{itemize}
\end{comment}
Given a basis \(\mathsf{B}\), the set of \textit{maximally resourceful states}, defined as \(\{\rho \mid \mathcal{T}_\mathsf{B}(\rho) = 1\}\), consists of all \textit{Fourier states} of the form \(\ket{s_j} = \frac{1}{\sqrt{d}} \sum_{k=1}^d \omega_d^{(k-1)(j-1)} \ket{k},\)
where \(1 < j \leq d\), \(\omega_d = \exp(2\pi \iota / d)\), and $\iota=\sqrt{-1}$. The set \(\{\ket{s_j}\}_{j=1}^d\) forms an orthonormal \textit{Fourier basis}, satisfying \(\braket{s_j}{s_k} = \delta_{jk}\). Therefore, any superposition or mixture of the basis states \(\{\ket{s_j}\}\) with \(j > 1\), results in a state that achieves the maximal texture value \(\mathcal{T}_\mathsf{B} = 1\).

As the texture of a density matrix depends on the chosen basis, we introduce below the notions of \emph{max-texture} and \emph{min-texture} associated with a quantum state \(\rho\).
\begin{definition}
    \label{def:max_min_texture}
    For a given quantum state $\rho$, the \textit{max-texture} and \textit{min-texture}, denoted by $\mathcal{T}^{\max}_{\mathbf{B}}(\rho)$ and $\mathcal{T}^{\min}_{\mathbf{B}}(\rho)$ respectively, are defined as the maximum and minimum values of the texture $\mathcal{T}_{\mathsf{B}}(\rho)$ taken over a specified set of orthonormal bases $\mathbf{B} = \{\mathsf{B}\}$. Formally,
\begin{equation}
\nonumber\mathcal{T}^{\max(\min)}_{\mathbf{B}}(\rho) = \underset{\mathbf{B}=\{\mathsf{B}\}}{\max(\min)} \mathcal{T}_{\mathsf{B}}(\rho).
\end{equation}
\end{definition}
\noindent
Alternatively, these extremal values can be expressed using unitary operators $U$ acting on a fixed reference basis $\widetilde{\mathsf{B}} \in \mathbf{B}$, such that the transformed basis $U\widetilde{\mathsf{B}} \in \mathbf{B}$. Mathematically, we can write
\begin{eqnarray}
\mathcal{T}^{\max(\min)}_{\mathbf{B}}(\rho) &=& 1 - \underset{\{U \mid U\widetilde{\mathsf{B}} \in \mathbf{B}\}}{\min(\max)} \bra{\tilde s_1} U\rho U^\dagger \ket{\tilde s_1},
\label{eq:def1}
\end{eqnarray}
where $\ket{\tilde s_1}$ denotes the only texture-less state with respect to the basis $\widetilde{\mathsf{B}}$. When $\mathbf{B}$ comprises all orthonormal bases in the Hilbert space, the max- and min-texture are denoted simply by $\mathcal{T}^{\max}(\rho)$ and $\mathcal{T}^{\min}(\rho)$, respectively.

%{Named this as min(max)-texture. We can slightly modify this definition as given a restricted set of bases.}\sout{For a given state $\rho$, the maximum and minimum values of $\mathcal{T}_{\mathsf{B}}$ are defined as $\mathcal{T}^{\max}(\rho)=\max_{\mathsf{B}}\mathcal T_\mathsf{B}(\rho)$ and $\mathcal{T}^{\min}(\rho)=\min_{\mathsf{B}}\mathcal T_\mathsf{B}(\rho)$ respectively. Note that, the maximization and minimization can be equivalently represented in terms of unitary matrices of dimension $d\times d$ as $\mathcal{T}^{\max}(\rho)= 1-\min_U \bra{s_1}U\rho U^\dagger\ket{s_1}$ and $\mathcal{T}^{\min}(\rho)= 1-\max_U \bra{s_1}U\rho U^\dagger\ket{s_1}$.}

\begin{lemma}
\label{lemma:max_min_texture}
For a given state \(\rho\), the max-texture and min-texture are given by \(\mathcal{T}^{\max}(\rho) = 1 - \lambda_d^\downarrow\) and \(\mathcal{T}^{\min}(\rho) = 1 - \lambda_1^\downarrow\), respectively. Here \(\lambda_1^\downarrow\) and \(\lambda_d^\downarrow\) denote the largest and smallest eigenvalues of the density matrix \(\rho\), respectively.
\end{lemma}
\begin{proof}
    Given a density matrix $\rho\in\mathbb C^d$, its spectral decomposition can be written as $\rho = \sum_{i=1}^d\lambda_i^{\downarrow}
    \ketbra{\psi_i}{\psi_i}$, where the eigenvalues $\{\lambda_i^\downarrow\}$ are arranged in descending order, i.e., \(1\geq\lambda_i^\downarrow\geq\lambda_{i+1}^\downarrow\geq0~\forall i\). The texture of the state \(\rho\) corresponding to a basis $\mathsf B$ can be written as
    \begin{eqnarray}
        \mathcal{T}_{\mathsf{B}}(\rho) &=& 1-\bra{s_1}\rho\ket{s_1}
        = 1-\sum_{i=1}^d\lambda_i^\downarrow|\braket{s_1}{\psi_i}|^2,
    \end{eqnarray}
    where $\ket{s_1}$ is the only texture-less state relative to the basis $\mathsf B$. Using Eq.~\eqref{eq:def1}, the expression for \(\mathcal{T}^{\max}(\rho)\) can be written as
    \begin{eqnarray}
        \nonumber\mathcal{T}^{\max}(\rho)&=&1-\min_U\bra{s_1}U\rho U^\dagger\ket{s_1}=1-\lambda_d^\downarrow,\nonumber
    \end{eqnarray}
   where the optimal unitary \(U_{\min}\) satisfies \(U_{\min}\ket{\psi_d} = \ket{s_1}\). Similarly, the minimum texture \(\mathcal{T}^{\min}(\rho)\) is given by \(\left(1 - \lambda_1^\downarrow\right)\), where the optimal unitary \(U_{\max}\) satisfies \(U_{\max}\ket{\psi_1} = \ket{s_1}\)~\cite{Biswas_PRA_2014}.
\end{proof}

\section{Quantification of purity through texture}
\label{sec:texture_purity}
Pure quantum states serve as valuable resources in numerous quantum communication and computation tasks. Historically, the majority of communication and computational protocols have been formulated and studied using pure states. Here, we introduce a texture-based measure for the purity of quantum states, thereby establishing a connection between texture and purity. 

In the resource theory of purity, the maximally mixed state is the only free state, while all pure states possess the same level of purity and represent the maximum resourceful states~\cite{Horodecki_PRA_2003, Gour_PhysicsReports_2015, Streltsov2018May, Luo_EPJD_2019}. The sets of free operations - mixture of unitaries $(\{\Lambda_{MU}\})$, noisy operations $(\{\Lambda_{NO}\})$, and unital operation $(\{\Lambda_U\})$ - form a hierarchical structure: $\{\Lambda_{MU}\}\subseteq\{\Lambda_{NO}\}\subseteq\{\Lambda_{U}\}$~\cite{Streltsov2018May}. Since unital maps constitute the largest set, any resource monotone under unital operations remains valid for the other two sets. Note that in the qubit scenario, all three sets are equivalent~\cite{Gour_PhysicsReports_2015}. 

On the other hand, in the resource theory of texture, we notice -- $(1)$ the maximally mixed state is the only state whose texture remains unchanged under any choice of basis; $(2)$ a pure state can be both resourceless or maximally resourceful as well, depending on the chosen basis. In order to connect texture and purity in a unified framework, we define purity monotone in terms of texture as 
\begin{equation}
    \mathcal{P}(\rho)=d\left(\mathcal{T}^{\max}(\rho)-\mathcal{T}^{\min}(\rho)\right).
    \label{eq:texture_based_purity}
\end{equation}
Here, \(d\) denotes the dimension of the system, and \(\mathcal{T}^{\max(\min)}\) are defined in the Definition~\ref{def:max_min_texture}, where \(\mathbf{B}\) encompasses all orthonormal bases in the \(d\)-dimensional Hilbert space.

\begin{theorem}
\label{th:purity_measure}
    The texture-based monotone of purity is a valid purity monotone under the set of unital operations.
\end{theorem} 
\begin{proof} For $\mathcal{P}(\rho)$ to qualify as a purity monotone, it must satisfy the following properties:
%-- $(1)$ \textit{faithful}: $\mathcal{P}(\rho)=0$ if and only if $\rho$ is maximally mixed state, i.e., $\mathbb{I}/d$, $(2)$ \textit{monotonicity}: $\mathcal{P}(\rho)$ must remain invariant or decrease monotonically under unital operations, i.e, $\mathcal{P}(\rho)\geq\mathcal{P}(\Lambda_U(\rho))~\forall\rho$, and $(3)$ \textit{convexity}: $\mathcal{P}(\rho)$ should be a convex function of density matrices, i.e., $\mathcal{P}(\sum_kp_k\rho_k)\leq\sum_kp_k\mathcal{P}(\rho_k)$. As mentioned earlier, the third condition is not strictly necessary, although a `good monotone' should satisfy this property in the context of convex resource theories. Next, we demonstrate that the texture-based purity monotone $\mathcal{P}(\rho)$ indeed satisfies these properties. 

$(1)$ \textit{Faithfulness}: From Lemma~\ref{lemma:max_min_texture}, one can express  
\begin{equation}
    \mathcal{P}(\rho)=d\left(\lambda_1^\downarrow(\rho)-\lambda_d^\downarrow(\rho)\right).
\end{equation}
Since the maximally mixed state is the only state in any dimension for which $\lambda_1^\downarrow=\lambda_d^\downarrow=\frac{1}{d}$, it follows that $\mathcal{P}(\mathbb{I}/d)=0$, while for all other states, $\mathcal{P}(\rho)>0$. Furthermore, pure states attain the maximum possible purity, which is equal to $d$.

$(2)$ \textit{Monotonicity}: By Uhlmann's theorem~\cite{Uhlmann_ROMP_1970,Nielsen_2012}, given two arbitrary density operators $\rho$ and $\sigma$, acting on same Hilbert space $\mathbb{C}^d$, $\Lambda_U(\rho)=\sigma\iff\rho\succ\sigma\iff\sum_{i=1}^n\lambda_i^\downarrow(\rho)\geq\sum_{i=1}^n\lambda_i^\downarrow(\sigma)~\forall ~n\in\{1,2,\cdots,d\}$. {The majorization condition (i.e., $\rho\succ\sigma$) thus consists of $d$ inequalities indexed by $n$. In particular, choosing $n=1$ yields $\lambda_1^\downarrow(\rho) \geq \lambda_1^\downarrow(\sigma)$, while setting $n=d-1$ implies $\lambda_d^\downarrow(\rho) \leq \lambda_d^\downarrow(\sigma)$, using the normalization $\sum_{i=1}^d \lambda_i(\rho)=\sum_{i=1}^d \lambda_i(\sigma)=1$.} Consequently, we obtain $(\lambda_1^\downarrow(\rho)-\lambda_d^\downarrow(\rho))\geq(\lambda_1^\downarrow(\sigma)-\lambda_d^\downarrow(\sigma))$. Hence, $\mathcal{P}(\rho)\geq\mathcal{P}(\Lambda_U\rho)~\forall\rho,~\text{and}~\Lambda_U$.

$(3)$ \textit{Convexity}: To establish the convexity of  $\mathcal{P}(\rho)$, we start with its definition,
\begin{eqnarray}
    \frac{1}{d}\mathcal{P}(\rho)&=&\mathcal{T}^{\max}(\rho)-\mathcal{T}^{\min}(\rho),\nonumber\\
    &=&\max_{\mathsf{B}}\mathcal{T}_{\mathsf{B}}(\rho)-\min_B\mathcal{T}_{\mathsf{B}}(\rho),\nonumber\\
    &=&\max_U\mathcal{T}_{\mathsf{B}_C}(U\rho U^\dagger)-\min_{U'}\mathcal{T}_{\mathsf{B}_C}(U'\rho U'^{\dagger}),\nonumber\\
    &=&\mathcal{T}_{\mathsf{B}_C}(U_{\max}^{\rho}\rho U_{\max}^{\rho\dagger})-\mathcal{T}_{\mathsf{B}_C}(U_{\min}^{\rho}\rho U_{\min}^{\rho\dagger}).
    \label{eq:convexity_1}
\end{eqnarray}
Here, $\mathsf{B}_C$ denotes the computational basis, and $U_{\max(\min)}^{\rho}$ is the unitary that attains the maximum (minimum) texture for the state $\rho$. Considering $\tilde{\rho}=\sum_k p_k\rho_k$, we can write
\begin{align}
   &\frac{1}{d}\mathcal{P}\left(\sum_k p_k\rho_k\right)\nonumber\\&=\mathcal{T}_{\mathsf{B}_C}(U_{\max}^{\tilde{\rho}}\tilde{\rho} U_{\max}^{\tilde{\rho}\dagger})-\mathcal{T}_{\mathsf{B}_C}(U_{\min}^{\tilde{\rho}}\tilde{\rho} U_{\min}^{\tilde{\rho}\dagger})\nonumber\\
   &=\mathcal{T}_{\mathsf{B}_C}\left(\sum_kp_kU_{\max}^{\tilde{\rho}}\rho_k U_{\max}^{\tilde{\rho}\dagger}\right)-\mathcal{T}_{\mathsf{B}_C}\left(\sum_kp_kU_{\min}^{\tilde{\rho}}\rho_k U_{\min}^{\tilde{\rho}\dagger}\right)\nonumber\\
   &=\sum_kp_k\Big[\mathcal{T}_{\mathsf{B}_C}(U_{\max}^{\tilde{\rho}}\rho_k U_{\max}^{\tilde{\rho}\dagger})-\mathcal{T}_{\mathsf{B}_C}(U_{\min}^{\tilde{\rho}}\rho_k U_{\min}^{\tilde{\rho}\dagger})\Big]\nonumber\\
   &\leq\sum_kp_k\Big[\mathcal{T}_{\mathsf{B}_C}(U_{\max}^{\rho_k}\rho_k U_{\max}^{\rho_k\dagger})-\mathcal{T}_{\mathsf{B}_C}(U_{\min}^{\rho_k}\rho_k U_{\min}^{\rho_k\dagger})\Big]\nonumber\\
   &=\frac{1}{d}\sum_kp_k\mathcal{P}(\rho_k),
   \label{eq:convexity_2}
\end{align}
where, in the third and fourth lines, we have utilized the fact that both the unitary transformation and $\mathcal{T}_{\mathsf{B}_C}$ act linearly on density operators. The final line follows directly from Eq. (\ref{eq:convexity_1}). This completes the proof.
\end{proof}

In the literature, there exists a family of purity monotones, known as \textit{R\'enyi} $\alpha$-purity~\cite{Gour_PhysicsReports_2015,Streltsov2018May}. This class of monotones is defined as $\mathbb{P}_\alpha(\rho)=\log_2d-S_\alpha(\rho)$, where $S_\alpha(\rho)=\frac{1}{1-\alpha}\log_2\left(\text{Tr}(\rho^\alpha)\right)$ is the R\'enyi $\alpha$-entropy of a density matrix \(\rho\) acting on a Hilbert space $\mathbb{C}^d$. {The quantity $\mathbb{P}_{\alpha}(\rho)$ admits an operational interpretation in the resource theories of purity. In particular, the quantity $\mathbb{P}_{\alpha}(\rho)$ reduces to the single-shot distillable purity in the limit $\alpha \to 0$, while in the opposite limit $\alpha \to \infty$, it corresponds to the single-shot purity cost. Moreover, the \textit{R\'enyi} $2$- purity $\mathbb{P}_2(\rho) = \log_2\left(d\, \mathrm{Tr}(\rho^2)\right)$ is a simple function of the linear purity $\mathrm{Tr}(\rho^2)$, and can be directly measured by interfering two identical copies of the quantum state  with each other~\cite{filip2002,Ekert_PRL_2002,Brun2004,Pichler_NJP_2013}. Below, we derive a lower bound on $\mathbb{P}_2(\rho)$ in terms of the texture-based purity monotone $\mathcal{P}(\rho)$. Conversely, an upper bound on $\mathcal{P}(\rho)$ can be expressed as a function of $\mathbb{P}_2(\rho)$.}

\begin{proposition}
\label{prop:purity_measures_bound}
In a \( d \)-dimensional Hilbert space \( \mathbb{C}^d \), the \textit{R\'enyi} $2$-purity \(\mathbb{P}_2(\rho)\) and the texture-based purity monotone \(\mathcal{P}(\rho)\) satisfy the inequality
\begin{equation}
    \label{eq:purity_inequality}
    \mathbb{P}_2(\rho) \geq \log_2\left[1 + \frac{\mathcal{P}(\rho)^2}{2d}\right].
\end{equation}
{Equality holds in the case of rank-$2$ density matrices which also incorporate density matrices acting on two-dimensional Hilbert space, i.e., \(\mathbb{C}^2\).}
\end{proposition}
\begin{proof}
    The \textit{R\'enyi} $2$-purity of a quantum state \( \rho \) acting on the Hilbert space \( \mathbb{C}^d \) is given by \( \mathbb{P}_2(\rho) = \log_2\left(d\, \mathrm{Tr}(\rho^2)\right) \). In contrast, the texture-based purity measure is defined as \( \mathcal{P}(\rho) = d\left(\lambda_1^\downarrow - \lambda_d^\downarrow\right) \).
    In order to establish the inequality ~\eqref{eq:purity_inequality}, we only need to show that
    \begin{equation}
        d\sum_{i=1}^d x_i^2\geq 1 + \frac{[d(x_1-x_d)]^2}{2d},
        \label{eq:ineq}
    \end{equation}
    where \(x_1\geq x_2\geq\cdots\geq x_d\geq 0\), and \(\sum_{i=1}^dx_i=1\). \\
    Let us interpret \( \{x_i\}_{i=1}^d \) as the possible values of a discrete random variable \( X \). Under the aforementioned constraints, the mean of \( X \) is \( 1/d \). The variance of \( X \) can be expressed as
    \begin{eqnarray}
        &&\frac{1}{d}\sum_i x_i^2-\frac{1}{d^2}=\frac{1}{d}\sum_i\left(x_i-\frac{1}{d}\right)^2,\nonumber\\
        \implies&& d\sum_i x_i^2=1+d\sum_i\left(x_i-\frac{1}{d}\right)^2.\nonumber
    \end{eqnarray}
   Now, since \( x_1 \geq \frac{1}{d} \geq x_d \), we have
  \begin{eqnarray}
    d\sum_i\left(x_i-\frac{1}{d}\right)^2&&\geq d\left[\left(x_1-\frac{1}{d}\right)^2+\left(\frac{1}{d}-x_d\right)^2\right],\nonumber\\
    && \geq d\frac{(x_1-x_d)^2}{2}\nonumber\\
    &&\geq d^2\frac{(x_1-x_d)^2}{2d},\nonumber
    \end{eqnarray}
    which establishes inequality~\eqref{eq:ineq}.\\
    Now, identifying \(x_i=\lambda_i^\downarrow\), the eigenvalues of \( \rho \), we obtain \(d\text{Tr}(\rho^2)\geq 1+\frac{\mathcal{P}(\rho)^2}{2d}\implies\mathbb{P}_2(\rho)\geq\log_2\left[1+\frac{\mathcal{P}(\rho)^2}{2d}\right]\), which proves the inequality~\eqref{eq:purity_inequality}.
    {Finally, it can be verified through straightforward algebra that the equality holds for rank-$2$ density matrices, which also include density matrices acting on $\mathbb{C}^2$.} This completes the proof.
\end{proof}

{Note that, in terms of distinguishability of states in $d>2$, $\mathbb{P}_2(\rho)$ is more capable than $\mathcal{P}(\rho)$. This is due to the fact that $\mathbb P_2(\rho)$ depends on the full spectrum of $\rho$ and is not one-to-one function of $\mathcal P(\rho)$. However, for $d=2$ they are equally powerful as evident from Proposition~\ref{prop:purity_measures_bound}.}

\begin{comment}
\textbf{Corollary $1$.} In $\mathbb{C}^2$, the texture-based purity monotone $\mathcal{P}(\rho)$ and the \textit{R\'enyi} $2$-purity are related by the equation $\mathbb{P}_2(\rho)=\log_2(1+\frac{1}{4}\mathcal{P}(\rho)^2)$.\\
\begin{proof}
  The \textit{R\'enyi} $2$-purity of a state $\rho$ is given by $\mathbb{P}_2(\rho)=\log_2\left(2\text{Tr}(\rho^2)\right)$. For an arbitrary qubit state $\rho$ with eigenvalues $\lambda_1$ and $\lambda_2$ (where $\lambda_1\geq\lambda_2$), $\mathbb{P}_2(\rho)=\log_2\left(2(\lambda_1^2+\lambda_2^2)\right)$, and $\mathcal{P}(\rho)=2(\lambda_1-\lambda_2)$. Using algebraic identities, we arrive at the relation $\mathbb{P}_2(\rho)=\log_2(1+\frac{1}{4}\mathcal{P}(\rho)^2)$.
\end{proof}
\textcolor{red}{Can there be any bound of \textit{R\'enyi} $\alpha$-purity imposed by textured-based purity?}
\end{comment}
{\textit{Experimental realization.} The estimation of $\mathcal P(\rho)$ can be performed following full state reconstruction of $\rho$ by standard tomographic methods \cite{Nielsen_2012}. Nonetheless, a substantial body of literature demonstrates that tomography is not strictly necessary -- the eigenvalue spectrum of $\rho$ can be obtained via direct spectrum-estimation protocols~\cite{Tanaka_PRA_2014, Lloyd_Nature_2014} as well as variational approaches, including variational quantum state diagonalization~\cite{LaRose_NPJ_QI_2019} and the variational quantum state eigensolver (VQSE)~\cite{Cerezo_NPJ_QI_2022}. In particular, the minimal and single-shot implementable setup can be executed by allowing multiple copies of $\rho$ to interfere -- an approach that has already been demonstrated in both linear optical and superconducting platforms~\cite{Tanaka_PRA_2014}, whereas VQSE requires only a single copy of the state $\rho$ per iteration, rendering them suitable for near-term quantum devices~\cite{Cerezo_NPJ_QI_2022}. Furthermore, by invoking the inequality in Eq.~(\ref{eq:purity_inequality}), an upper bound on the texture-based purity can be inferred from the \textit{R\'enyi} $2$-purity, which may be more readily accessible in experimental settings~\cite{filip2002,Ekert_PRL_2002,Brun2004,Pichler_NJP_2013}. However, this bound becomes weaker with growing system dimension. Additionally, any von Neumann measurement of $\rho$ in an arbitrary basis yields a lower bound on its actual purity, since nonselective von Neumann measurements correspond to unital channels. Finally, for states that are not full rank (i.e., $\lambda_d^\downarrow(\rho)=0$), the quantity $\lceil \log_2 \mathcal{P}(\rho) \rceil$ coincides with the single-shot purity cost~\cite{Streltsov2018May}.}

\section{Construction of monotones for convex resource theories based on texture}
\label{sec:texture_resource}
As discussed in the previous section, a resource theory is fundamentally characterized by three key components: the set of free states, the set of free operations, and a resource monotone that quantifies resource manipulation~\cite{chitambar_rmp_2019,horodecki_rmp_2009,streltsov_rmp_2017,Veitch2014Jan}. Consider a resource theory $\mathcal{R}$ defined by a set of free states $\mathcal{F_R}$ and a corresponding set of free operations $\mathcal{O_R}$. The theory $\mathcal{R}$ is said to be \textit{convex} if any convex combination of two free states remains free, i.e., $p\rho+(1-p)\sigma\in\mathcal{F_R}\;\forall\;\{\rho,\sigma\}\in\mathcal{F_R},\;0\leq p\leq 1$~\cite{horodecki_rmp_2009,baumgratz_prl_2014,Veitch2014Jan,Patra_PRA_2023}. Let us investigate how the texture properties of a quantum state relate to its resource content under a given convex resource theory $\mathcal{R}$ that satisfies certain conditions: $(i)$ $\mathcal{R}$ admits free pure states and free unitary operations, and $(ii)$ any pair of free pure states can be transformed into one another via a free unitary operation. Examples of such resource theories include coherence, non-stabilizerness, and entanglement, among others. {In particular, we demonstrate that for any convex resource theory $\mathcal{R}$ satisfying the aforementioned properties, all the free states possess vanishing min-texture with respect to a well-defined set of bases specifically tailored to $\mathcal{R}$. Recall that, for a given set of bases, the min-texture of a state is defined as the minimum amount of texture the state can exhibit across those bases.} To this end, we define a texture-based measure of the resource $\mathcal{R}$ for pure states as
\begin{equation}
    \mathcal{M}_\mathcal{R}(\ket{\psi})=\min_{U\in U_\mathcal{R}}\mathcal{T}_{\mathsf{B}_{\mathcal R}}(U\ket{\psi}),
    \label{eq:resource_texture}
\end{equation}
{where minimization is performed over the set $U_\mathcal{R}$ of free unitary operations defined by the resource theory $\mathcal{R}$. The texture is evaluated with respect to a resource-specific basis $\mathsf{B}_{\mathcal R}$. Although the choice of $\mathsf{B}_{\mathcal R}$ is not unique, it is constrained by the requirement that the texture-less state with respect to $\mathsf{B}_{\mathcal R}$ must be a free state within the resource theory $\mathcal{R}$. For instance, in non-stabilizerness theory, $\mathsf{B}_{\mathcal R}$ is naturally identified with the computational basis. Moreover, the set of free unitary operations $U_\mathcal{R}$ is fully specified by the underlying resource theory.} For example, in the resource theory of entanglement, $U_\mathcal{R}$ consists of all local unitaries of the form $U_1 \otimes U_2$; in the case of non-stabilizerness, it corresponds to the set of Clifford unitaries $U_C$; and for coherence, it includes incoherent unitaries $U_I$. Eq.~\eqref{eq:resource_texture} can be interpreted as follows. \textit{For a convex resource theory $\mathcal{R}$ satisfying the aforementioned properties, consider $\mathsf{B}_{\mathcal R}$ to be a basis whose corresponding texture-less state is a free state in $\mathcal{R}$. Let us define the set of bases $B_{\mathcal{R}} = \{U \mathsf{B}_{\mathcal R}:U\in U_{\mathcal{R}}\}$, consisting of all bases obtained by applying elements of $U_{\mathcal{R}}$ to $\mathsf{B}_{\mathcal R}$. For every free state in $\mathcal{R}$, there exists at least one basis in $B_{\mathcal{R}}$ with respect to which the state has vanishing texture. In contrast, any resourceful state necessarily exhibits nonzero texture in all such bases.}

We now show that $\mathcal{M}_{\mathcal{R}}(\ket{\psi})$ is a pure state resource monotone under the resource theory $\mathcal{R}$. A pure state resource monotone is a function that does not increase under any free operation that maps pure states to pure states (see e.g., Ref.~\cite{leone_pra_2024} for non-stabilizerness). Let $\tilde{\mathcal{O}}_{\mathcal{R}}$ denotes the collection of all such free operations. For $\mathcal{M_R}$ to qualify as a pure state resource monotone, it must satisfy the condition $\mathcal{M_R}\left(\Lambda(\ketbra{\psi}{\psi})\right) \leq \mathcal{M}_{\mathcal R}\left(\ketbra{\psi}{\psi}\right)$ for every operation $\Lambda \in \tilde{\mathcal{O}}_{\mathcal{R}}$. Note that, $U_{\mathcal R}$ is a subset of $\tilde{\mathcal{O}}_{\mathcal R}$. 
\begin{theorem}
    The texture-based monotone, $\mathcal{M_R}(\ket{\psi})$ is a valid pure state resource monotone within the convex resource theory $\mathcal{R}$, provided that $\mathcal{R}$ includes free pure states and free unitary operations, and that any two free pure states can be connected through a free unitary transformation.
    \label{th:measure}
\end{theorem}
\begin{proof}
The \textit{positivity} of $\mathcal{M_R}(\ket{\psi})$ follows from the fact that texture is always non-negative. Let us now consider the basis $\mathsf{B}_{\mathcal R}$ contains the states $\{\ket{{i^\mathcal{R}}}\}_{i=1}^d$, with the texture-less state being given by $\ket{s^\mathcal{R}_1}=\frac{1}{\sqrt{d}} \sum_{i=1}^d \ket{i^\mathcal{R}}$, which is also a free pure state in resource theory $\mathcal{R}$. Since the given resource theory is such that all the free pure states can be connected to $\ket{s^\mathcal{R}_1}$ through some free unitary, the monotone is \textit{faithful}. We now proceed to show that $\mathcal{M_R}(\ket{\psi})$ is non-increasing under free operations $\tilde{\mathcal{O}}_{\mathcal{R}}$ that map pure states to pure states. To begin, (by omitting the superscript $\mathcal{R}$ in $\ket{s^\mathcal{R}_1}$) we observe that~\cite{Nielsen_2012} \[ \mathcal{T}_{\mathsf{B}_{\mathcal R}}(U\ket{\psi}) = D_{\text{tr}}(U\ket{\psi}, \ket{s_1})^2\equiv D(U\ket{\psi}, \ket{s_1}), \] where \( D_{\text{tr}}(\ket{\psi}, \ket{\phi}) = \sqrt{1-|\braket{\psi}{\phi}|^2} \) is the trace distance between the two pure states $\ket{\psi}$ and $\ket{\phi}$. Let us now consider an arbitrary free operation \(\Lambda \in \tilde{\mathcal{O}}_{\mathcal{R}}\). For brevity, we denote \(\Lambda(\ketbra{\psi}{\psi}) \equiv \Lambda(\ket{\psi})\). Then, we can write
\begin{eqnarray}
    \mathcal{M_R}(\Lambda(\ket{\psi}))&=&\min_{U\in U_\mathcal{R}}\mathcal{T}_{\mathsf{B}_{\mathcal R}}\left(U\Lambda\left(\ket{\psi}\right)\right)\nonumber\\
    &=&\min_{U\in U_\mathcal{R}} D(U\Lambda(\ket{\psi}), \ket{s_1})\nonumber\\
    &=&\min_{U\in U_\mathcal{R}} D(\Lambda(\ket{\psi}), U^\dagger\ket{s_1})\nonumber\\
    &\leq&\min_{U\in U_\mathcal{R}} D(\Lambda(\ket{\psi}), \Lambda U^\dagger\ket{s_1})\nonumber\\
    &\leq&\min_{U\in U_\mathcal{R}} D(\ket{\psi},  U^\dagger\ket{s_1})\nonumber\\
    &=&\mathcal{M_R}(\ket{\psi}),
\end{eqnarray}
where the inequality in the fifth line uses that the trace distance is contractive under trace-preserving operations~\cite{Nielsen_2012}, and the fourth line follows from the fact that applying \(\Lambda\) to \(U^\dagger\ket{s_1}\) may reduce the cardinality of the set over which the minimization is performed. Hence proved.
\end{proof}

{\textbf{Remark.} Note that although convexity of a given resource theory is not required for the proof of Theorem~$3$, it plays a crucial role in ensuring the existence of a unique minimum distance in the definition of the texture-based resource measure when expressed as a distance-based quantity. Moreover, extending the texture-based pure state resource monotone to mixed states through the convex roof construction imposes convexity in the underlying resource theory. For this reason, we limit our analysis within the framework of convex resource theories.}

{\textit{Extending texture-based resource monotone for mixed states.} The geometric measure of resource within a convex resource theory is defined as \(\ \mathcal{M}^{G}_{\mathcal{R}}(\rho) = 1 - \max_{\sigma \in \mathcal{F_R}} F(\rho, \sigma)\), where the fidelity \(F(\rho, \sigma) = \|\sqrt{\rho} \sqrt{\sigma}\|_1^2\), with \(\|A\|_1 = \text{Tr}\left[\sqrt{A^\dagger A}\right]\) denotes the Schatten 1-norm~\cite{Shimony1995Apr, Wei_PRA_2003, Streltsov_NJP_2010} and the maximization is performed over the set of free states, \(\mathcal{F}_\mathcal{R}\). For the convex resource theories considered in \textbf{Theorem}~\ref{th:measure}, the geometric measure of resource for pure states coincides with the texture-based pure state resource monotone, i.e., $\mathcal{M}^{G}_{\mathcal{R}}(\ket{\psi})=\mathcal{M_R}(\ket{\psi})$. Consequently, under the class of free operations $\Lambda\in\mathcal{O_R}$ that cannot generate resource even probabilistically from free states, i.e., $\Lambda:\ket{\psi}\to\{p_i,\ket{\psi_i}$ with $\ket{\psi}\in\mathcal{F_R}\implies\ket{\psi_i}\in\mathcal{F_R}\forall i$, the texture-based pure state monotone satisfies the \textit{strong monotonicity} condition, $\mathcal{M_R}(\ket{\psi})\geq\sum_i p_i\mathcal{M_R}(\ket{\psi_i})$. This property allows one to extend the texture-based resource monotone to mixed states via the convex roof construction. The resulting mixed-state monotone, denoted by $\mathcal{M}^{\mathrm{CRE}}_{\mathcal{R}}(\rho)$, is defined as the minimum average of $\mathcal{M}_{\mathcal{R}}$ over all pure state decompositions of $\rho$~\cite{formation96,Wei_PRA_2003,Shimony1995Apr}. For resource theories in which the set of free states $\mathcal{F_R}$ is formed by the convex hull of free pure states, such as coherence, non-stabilizerness, and entanglement, $\mathcal{M}^{\text{CRE}}_{\mathcal{R}}(\rho)$ coincides with the quantity $\mathcal{M}^{G}_{\mathcal{R}}(\rho)$. Specifically,
\begin{equation}
\mathcal{M}^{G}_{\mathcal{R}}(\rho) = \mathcal{M}^{\text{CRE}}_{\mathcal{R}} (\rho) = \min_{\{p_i, \ket{\Psi_i}\}} \sum_i p_i \mathcal{M_R}(\ket{\Psi_i}),
\label{eq:mixed_state_monotone}
\end{equation} 
where the minimization is over all pure state decompositions $\{p_i, \ket{\Psi_i}\}$ of the mixed state $\rho$. This construction thus provides a natural extension of the texture-based resource monotone from pure states to mixed states.}

\subsection{Construction of  non-stabilizerness and coherence monotones for pure states}
\label{subsec:texture_magic_coherence}
We now present compact expressions of resource monotones for convex resource theories, such as non-stabilizerness, coherence, and entanglement, that satisfy the conditions of admitting free pure states and free unitary operations, with the additional property that any two free pure states can be connected via a free unitary transformation. Similar treatment for entanglement will be presented in the next subsection.

\textit{Non-stabilizerness monotone.} {The stabilizer formalism plays a crucial role in quantum information and computation, as many prototypical quantum error-correcting codes are built upon it~\cite{gottesman_pra_1996,knill2004,poulin_prl_2005}, and Clifford operations within this framework can be implemented fault-tolerantly~\cite{Bravyi_quantum_2019,litinski_quantum_2019}. Consequently, the concept of a resource theory for non-stabilizerness, commonly referred to as magic, has emerged, wherein stabilizer states are treated as free states and stabilizer operations incur no implementation cost~\cite{Veitch2014Jan,howard_prl_2017}.} Let us investigate how the texture properties of a quantum state relate to its degree of non-stabilizerness. To evaluate the pure state non-stabilizerness monotone, we take $\mathsf{B}_{\mathcal R}$ as the computational basis and $U_\mathcal{R}$ as the set of all Clifford unitaries. For a pure qubit state $\ket{\psi}$, the texture-based non-stabilizerness monotone is given by \[\mathcal{M}_{\mathcal{NS}}(\ket{\psi}) = \frac{1}{2}\left(1 - \max_{k\in \{x,y,z\}}|m_k|\right),\] where \( m_k = \text{Tr}(\sigma^k \ketbra{\psi}{\psi}) \) represents magnetization with \( \sigma^k~(k=x,y,z) \) being the Pauli operators. {Furthermore, the pure state non-stabilizerness monotone can be generalized to mixed states via the convex roof extension (Eq.~\eqref{eq:mixed_state_monotone}), as described previously.}

\textit{Coherence monotone.} {In contemporary quantum technologies, coherence has emerged as one of the fundamental resource, underpinning a wide range of quantum phenomena, including quantum interference~\cite{Daniel_PRA_2006, Prillwitz_PRA_2015, Bera_PRA_2015, Bagan_PRL_2016, Biswas_RSPA_2017}, quantum metrology~\cite{Zhang_PRL_2019, Ruvi_PRX_2024, Ares_OL_2021, Pires_PRA_2018}, entanglement~\cite{streltsov_prl_2015, Chitambar_PRL_2016, Streltsov_PRL_2016, Mekala_PRAL_2021}, phenomena of non-locality without entanglement~\cite{PATRA_PLA_2025}, and quantum communication~\cite{Kelly_SciPost_2023, Khan_PTRSA_2017, Shi_PRA_2017}.} Similar to quantum state texture, the coherence of a quantum state is a basis dependent quantity~\cite{streltsov_rmp_2017}. Specifically, for a given basis, the states that are diagonal in that basis are referred to as incoherent and constitute the set of free states in the resource theory of coherence. {Correspondingly, free operations are those quantum operations that preserve the set of incoherent states~\cite{streltsov_rmp_2017}. It is worth noting that several distinct classes of free operations arise within this framework. The most general class consists of all trace-preserving, completely positive, and non-selective quantum operations that leave the set of incoherent states invariant, known as \emph{maximally incoherent operations} (MIO)~\cite{aberg_06}. A more restrictive, yet operationally motivated, class is that of \emph{incoherent operations} (IO), which are defined by the stronger requirement that coherence cannot be generated even probabilistically from an incoherent input state~\cite{baumgratz_prl_2014}.} To compute texture-based measure of coherence in the computational basis, we consider $\mathsf{B}_{\mathcal R}$ as the Fourier basis and define $U_\mathcal{R}$ as the set of unitaries that do not increase coherence, such as Heisenberg-Weyl unitaries (Pauli matrices for qubits) and any unitary that is diagonal in the computational basis. For an arbitrary pure qudit state, $\ket{\psi}=\sum_{i=0}^{d-1}c_i\ket{i}$, the coherence in computational basis $\{\ket{i}\}$ is quantified by \[\mathcal{M}_{\mathcal{C}}(\ket{\psi})=1-\max\{|c_0|^2,|c_1|^2,\cdots,|c_{d-1}|^2\},\] where $c_i=\langle i|\psi\rangle$. As before, the pure state coherence monotone admits a natural extension to mixed states via Eq.~(\ref{eq:mixed_state_monotone}). {Both the pure state monotone and its convex-roof extension to mixed states are monotonic under MIO. Furthermore, they satisfy the stricter condition of strong monotonicity under IO~\cite{streltsov_rmp_2017}.}

\subsection{Construction of entanglement monotone}
\label{subsec:texture_entanglement}
Entanglement \cite{horodecki_rmp_2009} is a fundamental resource for achieving quantum advantage in various tasks, including quantum communication and quantum key distribution. Within the resource theoretical framework, local operations and classical communication (LOCC) are considered as free operations, while all separable states are regarded as free states. Over the decades, extensive research has been conducted on detecting and quantifying entanglement, firmly establishing its theoretical foundation. Here we explore the relationship between the texture properties of a quantum state and the entanglement shared between parties. To begin with, let us introduce the notion of \emph{non-local texture}. 

\begin{definition}
\label{def:non-local_texture}
The non-local texture is defined as the minimum texture that a state can possess across the local bases. Mathematically, for a bipartite density matrix $\rho_{A_1A_2}$ acting on $\mathbb{C}^{d_{A_1}}\otimes\mathbb{C}^{d_{A_2}}$ (hereafter abbreviated as $d_{A_1}\otimes d_{A_2}$), it reads as 
\begin{eqnarray}
&&\mathcal{T}_{NL}^{A_1:A_2}\left(\rho_{A_1A_2}\right)\nonumber\\&&=\min_{\{p_i,\ket{\Psi_i}\}}\sum_i p_i \min_{\{{\mathsf B}_{A_1}\otimes {\mathsf B}_{A_2}\}}\mathcal{T}_{{\mathsf B}_{A_1}\otimes {\mathsf B}_{A_2}}\left(\ket{\Psi_i}_{A_1A_2}\right)\nonumber\\&&=\min_{\{p_i,\ket{\Psi_i}\}}\sum_ip_i\mathcal{T}^{\min}_{\mathbf{B}_{A_1}\otimes \mathbf{B}_{A_2}}\left(\ket{\Psi_i}_{A_1A_2}\right),
\label{eq:non-local_texture}
\end{eqnarray}
where $\rho_{A_1A_2}=\sum_ip_i\ketbra{\Psi_i}{\Psi_i}$ is a pure state decomposition of $\rho_{A_1A_2}$, and $\mathcal T^{\min}_{\mathbf{B}_{A_1}\otimes \mathbf{B}_{A_2}}$ is the minimum texture after optimizing over the set of all product bases in the Hilbert space $d_{A_1}\otimes d_{A_2}$, denoted by  $\mathbf{B}_{A_1}\otimes \mathbf{B}_{A_2}=\{\mathsf B_{A_1}\otimes \mathsf B_{A_2}\}$. The outer minimization is performed over all possible pure state decompositions of $\rho_{A_1A_2}$. 
% For an arbitrary pure state $\ket{\Psi}_{A_1A_2}$, Eq.~\eqref{eq:non-local_texture} reduces to a simplified expression,
% \begin{equation}
% \mathcal{T}_{NL}^{A_1:A_2}\left(\ket{\Psi}_{A_1A_2}\right)=\mathcal{T}^{\min}_{\mathbf{B_{A_1}\otimes B_{A_2}}}\left(\ket{\Psi}_{A_1A_2}\right). 
% \label{eq:non-local_texture_pure}
% \end{equation}
\end{definition}

Notice that, similar to the texture-based resource monotone as defined in Eq.~\eqref{eq:resource_texture}, a non-vanishing non-local texture of a given state implies that the state possesses non-zero amount of texture with respect to all local bases. By identifying $\mathsf{B}_{\mathcal R}$ as an arbitrary product basis of the form $\mathsf{B}_{A_1} \otimes \mathsf{B}_{A_2}$ and $U_\mathcal{R}$ as the set of all local unitaries of the form $U_{A_1} \otimes U_{A_2}$, {it follows that, for any pure bipartite state $\ket{\Psi}_{A_1A_2}\in d_{A_1}\otimes d_{A_2}$, Eq.~\eqref{eq:resource_texture} reduces to $\mathcal{T}_{NL}^{A_1:A_2}\left(\ket{\Psi}_{A_1A_2}\right)=\mathcal{T}_{\mathbf{B}_{A_1}\otimes \mathbf{B}_{A_2}}^{\min}\left(\ket{\Psi}_{A_1A_2}\right)$. Similarly, for a mixed state $\rho_{A_1A_2}$, Eq.~\eqref{eq:mixed_state_monotone} coincides with Eq.~\eqref{eq:non-local_texture}.} Therefore, by Theorem~\ref{th:measure}, together with the fact that the set of separable states is the convex hull of pure product states, we arrive at the following corollary.

\begin{corollary}
\label{cor:non-local_texture_and_entanglement}
The non-local texture serves as a valid entanglement monotone for any bipartite state $\rho_{A_1A_2}$.   
\end{corollary}

A bipartite pure state in the Hilbert space $d_{A_1}\otimes d_{A_2}$, where $d_{A_1} \leq d_{A_2}$, can be expressed via the Schmidt decomposition as $\ket{\Psi}_{A_1A_2} = \sum_{i=1}^{d_{A_1}} \sqrt{\lambda_i} \ket{i}\ket{i'}$. Here, the Schmidt coefficients $\{\lambda_i\}$ are ordered such that $0 \leq \lambda_{i+1} \leq \lambda_i \leq 1$ and $\sum_i \lambda_i = 1$, while $\{\ket{i}\}$ and $\{\ket{i'}\}$ form orthonormal sets, known as the Schmidt basis. The texture-based entanglement measure (alternatively, the non-local texture) for an arbitrary pure bipartite state can take the compact form as 
\begin{equation} 
\mathcal{T}_{NL}^{A_1:A_2}\left(\ket{\Psi}_{A_1A_2}\right) = \mathcal{M}_{\mathcal{E}}(\ket{\Psi}_{A_1A_2}) = 1 - \lambda_1,
\end{equation}
with $\lambda_1$ being the largest Schmidt coefficient~\cite{vidal_prl_1999,horodecki_rmp_2009}. For a mixed bipartite state $\rho_{A_1A_2}$, this generalizes to 
\begin{align} 
\nonumber \mathcal{T}_{NL}^{A_1:A_2}\left(\rho_{A_1A_2}\right)&= \mathcal{M}_{\mathcal{E}}^{\text{CRE}}(\rho_{A_1A_2}) \\& =\min_{\{p_i,\ket{\Psi_i}\}}\sum_i p_i(1-\lambda_1^{\ket{\Psi_i}}),
\end{align}
where the minimization is performed over all possible pure state decompositions of $\rho_{A_1A_2}=\sum_ip_i\ketbra{\Psi_i}{\Psi_i}$, and $\lambda_1^{\ket{\Psi_i}}$ represents the largest Schmidt coefficient of the constituent state $\ket{\Psi_i}$.

\textbf{Multipartite entanglement characterization.} The notion of non-local texture can be extended to the multipartite scenario by choosing either the set of fully local bases of the form $\{B_{A_1}\otimes B_{A_2}\otimes\cdots\otimes B_{A_N}\}$ or the set of bi-local basis of the form $\{B_{A_1A_2\cdots A_i}\otimes B_{A_{i+1\cdots N}}\}$. 
%within Eqs.~\eqref{eq:non-local_texture}, and \eqref{eq:non-local_texture_pure}.
A multipartite state that exhibits non-zero texture with respect to all fully local bases is said to possess \emph{multipartite non-local texture}, whereas exhibiting texture across all bi-local bases indicates the presence of \emph{genuine multipartite non-local texture}. Recall that a multipartite state $\ket{\Psi}_{A_1A_2\cdots A_N} \in \mathbb{C}^{d_{A_1}} \otimes \mathbb{C}^{d_{A_2}} \otimes \cdots \otimes \mathbb{C}^{d_{A_N}}$ is said to be entangled if it is entangled across at least one bipartition, and is termed genuinely multipartite entangled if it is entangled across every possible bipartition. Accordingly, a texture-based measure of \emph{multipartite entanglement} (ME) can be defined by choosing $\mathsf{B}_{\mathcal R}$ as an arbitrary fully local basis, and $U_\mathcal{R}$ as the set of all fully local unitaries of the form $U_1 \otimes U_2 \otimes \cdots \otimes U_N$. On the other hand, to define a texture-based measure of \emph{genuine multipartite entanglement} (GME), we instead set $U_\mathcal{R}$ to be the set of all bi-local unitaries, such as $U_{12\cdots i} \otimes U_{i+1\cdots N}$ \(\forall i\). Upon evaluation, the texture-based measure of genuine multipartite entanglement for an arbitrary state $\ket{\Psi}_{A_1A_2\cdots A_N}$ reduces to (cf. \cite{ggm10})
\begin{equation}
\mathcal{M}_{\mathcal{GME}}(\ket{\Psi}) = 1 - \max_{\mathcal{A:\bar{A}}} \lambda_{1(\mathcal{A:\bar{A}})},
\label{eq:GME_texture}
\end{equation}
where $\lambda_{1(\mathcal{A:\bar{A}})}$ denotes the largest Schmidt coefficient in the bipartition $\mathcal{A:\bar{A}}$, and the maximization is taken over all possible bipartitions of the system. It is noteworthy that $\mathcal{M}_{\mathcal{GME}}(\ket{\Psi})$ coincides exactly with the \emph{generalized geometric measure} (GGM) of genuine multipartite entanglement for pure states~\cite{ggm10} and can be extended to mixed states~\cite{das_pra_2016}.

\section{Characterizing phase transitions using texture-based metrics}
\label{sec:texture_QPT}
One significant advantage of the texture measure lies in its experimental accessibility and computabality. For instance, while the entanglement~\cite{osborne_pra_2002}, discord-based measures~\cite{raoul_prb_2008} can capture phase transition, their evaluation remains challenging in laboratories as it requires full tomography of the state. In contrast, the texture measure proposed in \cite{Parisio2024}, referred to as rugosity, offers a more experimentally accessible alternative that can be obtained without tomography. In particular, it is defined for a pure quantum state \(\ket{\Psi}\) as
\begin{equation}
    \mathfrak{R}(\ket{\Psi}) = -\ln\left[|\overline{\langle{+}|}\Psi\rangle|^2\right],
\end{equation}
where the texture is evaluated in the computational basis. Operationally, rugosity corresponds to evaluating the probability amplitude of the quantum state \(\ket{\Psi}\) projected onto the product state \( \overline{\ket{+}} = \ket{+}^{\otimes N} \) with \( \ket{+} \) being the eigenstate of \( \sigma^x \) with eigenvalue \( +1 \). 
%This makes the texture experimentally tractable, as it can be measured in various platforms such as ion-trap and superconducting qubit systems using measurements in the computational basis. 
Given this formulation, it is interesting to explore the potential of using the texture to detect quantum phase transitions, which we demonstrate next. 
{Note that although choosing a basis other than the computational basis may lead to different numerical values of $\mathfrak{R}$, this does not alter the qualitative behavior relevant for detecting critical points in the phase diagram. In particular, the qualitative behavior of the curvature in the vicinity of the critical point remains unchanged, even though the numerical values of $\mathfrak{R}$ may vary. Consequently, the identification of the phase transition points is unaffected. A similar kind of observation has also been reported in the studies of detecting quantum phase transitions using quantum coherence measures \cite{chen_pra_2016,Li2016May,li_prb_2020}.}
%In this work, we demonstrate that the ground state texture exhibits a curvature change, when plotted against some relevant parameter, near the phase transition point in a spin chain model.
%\textcolor{black}{Note that the measure of texture, \(\mathfrak R(\ket{\Psi})\), is basis-dependent, although the behavior at the quantum critical point remains unaffected for almost all choices of basis
%. This behavior is similar to one that was observed when coherence is used as an indicator of a phase transition~
%(cf. \cite{chen_pra_2016,Li2016May,li_prb_2020} for coherence).}
%, which also depends upon the choice of basis.} 

\subsection{Ising spin chain with longitudinal magnetic field} 

Let us consider the ground state of an Ising spin chain~\cite{sachdev_2011} of length $N$, governed by the Hamiltonian
\begin{equation}
    H = -\left[\sum_{j=1}^{N} \frac{J'}{4} \sigma_j^x \sigma_{j+1}^x + \frac{h'}{2} \sum_{j=1}^{N} \sigma_j^z + \frac{g'}{2} \sum_{j=1}^{N} \sigma_j^x \right],
    \label{eq:Hamiltonian}
\end{equation}
where \(\sigma^k\) (\(k = x, y, z\)) are the Pauli matrices, \(J'\), \(h'\) and \(g'\) represents the strength of the nearest-neighbor spin interaction, the strength of the transverse magnetic field, and the longitudinal magnetic field respectively. We introduce dimensionless parameters such as \(h = h'/J'\) and \(g = g'/J'\) with considering periodic boundary conditions (PBC), i.e., \(\sigma_{N+1} \equiv \sigma_1\). Note that the nonvanishing \(g\) prohibits the exact solvability of the model. Accordingly, we examine two distinct cases: the integrable regime with \(g = 0\), and the non-integrable regime with \(g \ne 0\). When \( h < 1 \) and \( g < 0 \), the ground state belongs to the ferromagnetic down phase, while for \( h < 1 \) and \( g > 0 \), it belongs to the ferromagnetic up phase~\cite{Yuste2018Apr}. The line \( g = 0 \) corresponds to a first-order phase transition between these two ferromagnetic phases. For \( h > 1 \), \(g=0\), the ground state lies in the paramagnetic phase, and \( h = \pm 1 \) represents the second-order quantum phase transition points separating the ferromagnetic and paramagnetic regions.

\begin{figure}
    \centering
    \includegraphics[width=\linewidth]{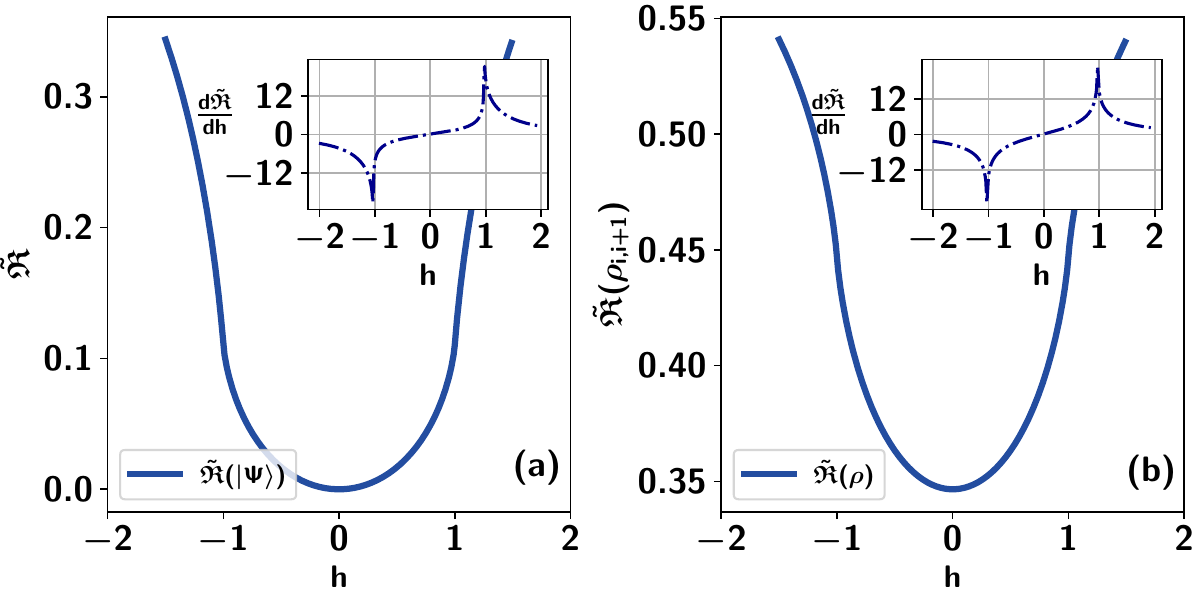}
    \caption{\textbf{Normalized rugosity of the ground state against the external magnetic field, \(h\).} Normalized value of the texture (ordinate) of the ground state (panel (a)) and its reduced two-party density matrix (panel (b)) for the Ising chain,  as functions of the transverse field $h$ (abscissa). The plots exhibit symmetry about \(h = 0\) and display distinct curvature changes near the phase transition points at \(h = \pm 1\), highlighted by the derivative \(\frac{d\mathfrak{R}}{dh}\) shown in the insets. Here \(g = 0\) and \(N = 512\). All axes are dimensionless.}
    \label{fig:ising_tex}
\end{figure}
\emph{Vanishing longitudinal magnetic field \((g=0)\).} In this case, the model becomes analytically solvable by mapping the spin degrees of freedom to free fermions via the Jordan-Wigner transformation \cite{barouch_pra_1970_1, barouch_pra_1970_2, lieb1961}, followed by a Fourier transformation, allowing the diagonalization of the Hamiltonian. 
\begin{comment}
Accordingly, the various phases of the ground state \( \ket{\Psi} \) can be characterized by the texture function, given by
\begin{equation}
    \mathfrak{R}(\ket{\Psi}) = -\sum_{p=1}^{N/2} \log\left[\text{abs}\left(v_p \cos\frac{\phi_p}{2} - i u_p \sin\frac{\phi_p}{2} \right)^2\right],
\end{equation}
where \(\phi_p = (2p - 1)\frac{\pi}{N}\), \(u_p = \cos\theta_p\), and \(v_p = i\sin\theta_p\), where \(\theta_p\) is the Bogoliubov angle defined as \(\theta_p = \cos^{-1}\left[\frac{-h + \cos\phi_p - \lambda}{\sqrt{2(\lambda^2 - (\cos\phi_p - h)^2)}}\right]\), with \(\lambda = \sqrt{(h - \cos\phi_p)^2 + (\gamma \sin\phi_p)^2}\). 
\end{comment}
{Accordingly, the various phases of the ground state \( \ket{\Psi} \) can be characterized by the texture function, given by
\begin{equation}
    \mathfrak{R}(\ket{\Psi}) = -\sum_{p=1}^{N/2} \ln\left[\sin^2 \left(\theta_p - \frac{\phi_p}{2}\right)\right],
\end{equation}
where \( \phi_p = \frac{(2p - 1)\pi}{N} \), and \( \theta_p \) denotes the Bogoliubov angle, defined as \( \theta_p = \cos^{-1}\left( \frac{\cos\phi_p -h - \lambda }{ \sqrt{ 2\left( \lambda^2 - (\cos\phi_p - h)^2 \right) } } \right) \), with \( \lambda = \sqrt{ (h - \cos\phi_p)^2 + \sin^2\phi_p } \).} A detailed derivation can be found in Appendix~\ref{sec:xy_spin_chain}. As the rugosity increases with the increment of system size, we use normalized rugosity given as \(\tilde{\mathfrak{R}}=\mathfrak{R}/N\) with \(N\) being the system size. In Fig.~\ref{fig:ising_tex}(a), we plot the normalized rugosity of the ground state by varying the external magnetic field \(h\) with \(N=512\). The plot reveals that \(\tilde{\mathfrak{R}}(\ket{\Psi})\) vanishes when the external field is zero. As the external field increases, the texture grows and eventually saturates to a finite value. This behavior can be understood by analyzing the phases of the ground state of the model. Specifically, in the limit \(h \sim 0\), the ground state approximates an eigenstate of the operator \(\sum_j \sigma_j^x \sigma_{j+1}^x\), which is close to the product state \(\overline{\ket{+}}\). Consequently, the texture vanishes at \(h \to 0\). In contrast, in the regime \(h \gg 1\), the ground state tends toward that of \(\sum_j \sigma_j^z\), corresponding to a paramagnetic phase. In this regime, the texture saturates to a finite value. Therefore, as \(h\) increases from zero, the texture of the ground state rises and saturates. A notable aspect of this growth is the distinct change in curvature near \( h = \pm 1 \), which can be captured by the kink in \(\frac{d\mathfrak{R}}{dh}\), thereby detecting the critical point of the spin chain (see the inset of Fig.~\ref{fig:ising_tex}(a)). \textcolor{black}{
So far, we have calculated the rugosity of the entire ground state, which may not always be practically feasible. In such cases, analyzing subsystem properties can be insightful for detecting phase transitions. For instance, quantities like bipartite-entanglement~\cite {Osterloh2002Apr,osborne_pra_2002} and quantum discord~\cite{raoul_prb_2008}, when computed for a two-party reduced subsystem, often indicate critical points in the system. Motivated by this observation, we now focus on the rugosity of a two-party reduced subsystem within the full $N$-party system. In particular, we compute the rugosity of a subsystem of two nearest-neighbor parties, which can be efficiently evaluated by exploiting the symmetry of the Hamiltonian~\cite{osborne_pra_2002,sende_pra_2005,mishra_pra_2013,sadhukhan_pre_2016}. Furthermore, under periodic boundary conditions, all such two nearest-neighbor reduced density matrices parties are equivalent,} \textcolor{black}{ given by
\begin{figure}
    \centering
    \includegraphics[width=\linewidth]{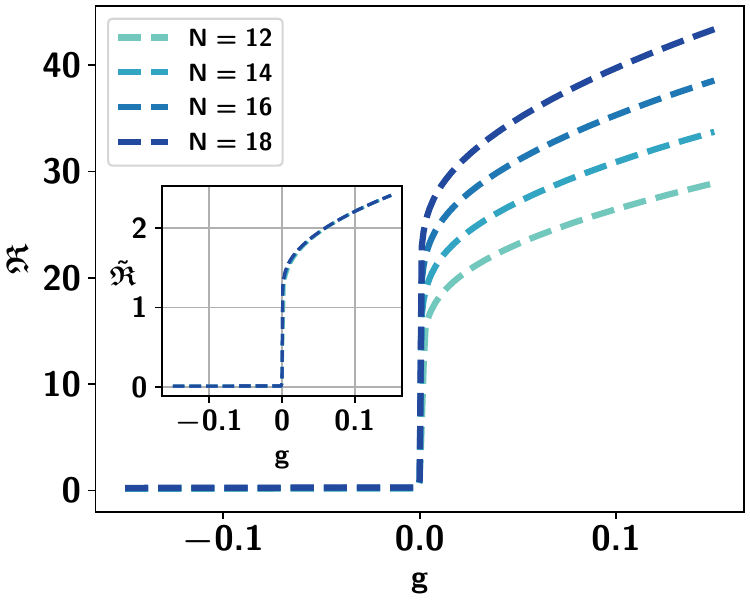}
    \caption{{\textbf{Rugosity (\(\mathfrak{R}\)) (ordinate) of the ground state vs the longitudinal magnetic field, \(g\) (abscissa).} We present the texture of the ground state for an Ising chain subjected to a longitudinal and transverse magnetic field with different system sizes (normalized rugosity is plotted in the inset). The texture remains zero for \(g < 0\) and increases with \(g\) for \(g \geq 0\), capturing the phase transition from the paramagnetic to the ferromagnetic phase at the critical point \(g = 0\). We set \(h = 0.5\). All axes are dimensionless.}}
    \label{fig:texture_hx}
\end{figure}
\begin{equation}
    \rho_{i,i+1}=\frac{1}{4}\bigg(\mathbb{I}+m_z(\sigma_i^z+\sigma_{i+1}^z)+\sum_{\alpha\in\{x,y,z\}} C_{\alpha\alpha}\sigma_i^\alpha\sigma_{i+1}^\alpha\bigg),
    \label{eq:two_party_reduced}
\end{equation}
where \(m_z\) denotes the magnetization along the \(z\)-direction, and \(C_{\alpha\alpha}\) represents the spin-spin correlations in the \(x\), \(y\), and \(z\) directions, respectively. The rugosity of this state, when expressed in the computational basis, can be straightforwardly computed as
\begin{equation}
    \mathfrak{R}(\rho_{i,i+1})=-\ln\left[\frac{1+C_{xx}}{4}\right].
\end{equation}
We observe that the rugosity of the state given in Eq.~\eqref{eq:two_party_reduced} depends solely on the correlation function in the \(x\)-direction. This dependence arises from the structure of the density matrix. Since \(C_{xx}\) exhibits a change in curvature near the phase transition, the rugosity also reflects such a change (as demonstrated in Fig. \ref{fig:ising_tex}(b)) near the critical point. Therefore, the texture of the two-party reduced subsystems obtained from the ground state can also serve as an effective indicator of phase transitions, and is experimentally more feasible to analyze.
}
Furthermore, to demonstrate that this phenomenon is not limited to integrable models, we consider a non-integrable spin chain model obtained by introducing a longitudinal magnetic field, i.e., by setting \( g \ne 0 \) in the Hamiltonian defined in Eq.~\eqref{eq:Hamiltonian}.

\emph{Nonvanishing longitudinal magnetic field \((g\neq0)\).} The texture \( \mathfrak{R} \) of \(N\) party ground state remains nearly vanishing for \( g \le 0 \),  while it becomes appreciably non-zero and increases with \(g>0\) having a sharp increase at \(g=0\) and continues to grow with \(g>0\). This behavior indicates that the texture of the ground state effectively captures the transition point. Hence, the texture serves as a useful indicator of phase transitions in both integrable and non-integrable models.

In Fig.~\ref{fig:texture_hx}, we plot the texture of the ground state corresponding to \(g\ne 0\). The figure demonstrates that the texture \( \mathfrak{R} \) remains nearly vanishing for \( g < 0 \), but as \( g \) crosses \( g = 0 \), it becomes appreciably nonvanishing and shows a sharp increase. This behavior indicates that the texture of the ground state effectively captures the transition point. Hence, the texture serves as a useful indicator of phase transitions in both integrable and non-integrable models.

% \section{Circuit characterization task}
% \label{sec:circuit_char}
% In \cite{Parisio2024} they have shown that in a single layer of the circuit, among the universal set of gates $\{U_H, U_S, U_T, U_{CNOT}\}$, by average texture one can detect the position of qubits where $CNOT$ gate is applied. \textbf{By other resources like purity, imaginarity, coherence etc along with texture, can we also detect the single qubit gates??}

\section{Conclusion}
\label{sec:conclusion}
Evaluating the significance of the recently defined texture resource theory  requires establishing the connection between quantum state texture and other established resource theories.  By introducing the concepts of maximum and minimum texture with respect to a given set of bases, we demonstrated that the purity of a quantum state is proportional to their difference, which is monotone under unital operations, and  provides a lower bound of the R\'enyi entropy-based purity. 

We proposed a general framework of texture-based monotone developed for convex resource theories with a set of free pure states that are closed under free unitary operations. Additionally, we obtained the expression of these monotones for pure states in the case of non-stabilizerness and coherence. 
Its connection with the entanglement theory was established by showing that a bipartite state is entangled if and only if it has a non-vanishing non-local texture, which is the minimum texture that the state can have across the local bases.  The non-local texture, which has been proven to be an entanglement monotone, is quantitatively equivalent to the geometric measure of entanglement in the multipartite domain and  can be extended for mixed states through the convex roof extension.

We found that, in the presence of longitudinal and transverse magnetic fields, quantum phase transitions in the Ising chain may be detected by the texture of the ground state with respect to almost any basis.  The simplicity lies in the fact that the texture of reduced subsystems of the ground state can also reveal the quantum critical point.  As the texture can be experimentally assessed without complete state tomography, our findings suggests that an uneven (i.e., non-uniform) landscape of density matrix elements might not only effectively bridge diverge quantum resources but also serve as a useful tool for detecting  critical phenomena  in quantum many-body systems.

\acknowledgements

We acknowledge the use of \href{https://github.com/titaschanda/QIClib}{QIClib} --  a modern C++ library designed for general-purpose quantum information processing and quantum computing. We also thank the cluster computing facility at the Harish-Chandra Research Institute for computational support. We acknowledge support from the project entitled ``Technology Vertical - Quantum Communication'' under the National Quantum Mission of the Department of Science and Technology (DST)  (Sanction Order No. DST/QTC/NQM/QComm/$2024/2$ (G)). This research is partially supported by the INFOSYS scholarship for senior students.

\bibliography{ref}

\appendix
\section{Ising-spin chain model diagonalization}
\label{sec:xy_spin_chain}

In this work, we primarily focus on the ground state of an Ising spin chain, described by the Hamiltonian:

\begin{eqnarray}
    H &=& -\Bigg [\sum_{j=1}^{N} \frac{J'}{4} \sigma_j^x\sigma_{j+1}^x   + \frac{h'}{2}\sum_{j=1}^{N}\sigma_j^z\Bigg ].
\end{eqnarray}
This model can be solved analytically by mapping the spins to free fermions using the Jordan-Wigner transformation, which is given as
\begin{align}
    \sigma^x_n &= \left( c_n + c_n^\dagger \right) \prod_{m<n}(1-2 c_m^\dagger c_m) \\
    \sigma^y_n &= i\left( c_n - c_n^\dagger \right) \prod_{m<n}(1-2 c_m^\dagger c_m) \\
    \sigma^z_n &= 1 - 2 c_n^\dagger c_n,
 \label{eq:Jordan_wigner}
\end{align}
where \(c_m^\dagger\) and \(c_m\) are the fermionic creation and annihilation operators, respectively, obeying the standard fermionic commutation relations. Using this transformation, the Hamiltonian can be expressed in its free fermionic form as
\begin{eqnarray}
   H &=& \nonumber\sum_{n}\frac{1}{2}\left (c_n^\dagger c_{n+1} + c_n^\dagger c_{n+1}^\dagger + \text{h.c.} \right) + h\left(c_n^\dagger c_n - \frac{1}{2}\right).\\&&
   \label{eq:JW_hamil}
\end{eqnarray}
The Hamiltonian in Eq.~\eqref{eq:JW_hamil} can be diagonalized via a Fourier transform, enabling the computation of various two-point correlators in the fermionic basis. To proceed, we apply the Fourier transformation, which maps the fermionic operators to their corresponding conjugate momentum modes as follows:
\begin{eqnarray}
    \nonumber c_j &=& \frac{1}{\sqrt{N}}\sum_{p=-N/2}^{N/2}\exp\left(-\frac{2\pi jp}{N}\right)a_p, \\ 
    \text{and} \quad c_j^\dagger &=& \frac{1}{\sqrt{N}}\sum_{p=-N/2}^{N/2}\exp\left(\frac{2\pi jp}{N}\right)a_p^\dagger.
\end{eqnarray}
Imposing periodic boundary conditions imparts translational invariance to the system, making momentum a good quantum number. This allows the Hamiltonian to be decomposed into independent momentum sectors, expressed as \(H_{JW} = \otimes H_p\), where each \(H_p\) is given by
\begin{eqnarray}
    \nonumber H_p &=& \sum_{p>0} (h + \cos\phi_p)(a_p^\dagger a_p + a_{-p}^\dagger a_{-p}) \nonumber\\
    && + \sin\phi_p \left[ a_p^\dagger a_{-p}^\dagger + a_p a_{-p} \right] - h, \nonumber
    \label{pksea_xy}
\end{eqnarray}
with \(\phi_p = (2p - 1)\pi / N\) and \(p \in \{1, \ldots, N/2\}\). We adopt anti-periodic boundary conditions, i.e., \(c_{N+1} = -c_1\), as discussed in \cite{santoro_ising_beginners_2020}. In the thermodynamic limit \(N \to \infty\), the momentum becomes continuous, leading to \(\phi_p \in (0, \pi)\). In the basis \(\{|0\rangle, a^\dagger_p a^\dagger_{-p}|0\rangle\}\), the Hamiltonian \(H_p\) takes the form
\begin{equation}
    H_p = \left[\begin{array}{cc}
    -h - \cos \phi_p & -i \sin \phi_p \\
    i \sin \phi_p & \cos \phi_p + h
    \end{array}\right].
    \label{eq:xy_P}
\end{equation}
This Hamiltonian can be diagonalized using a Bogoliubov transformation:
\begin{equation}
    \begin{pmatrix}
        c_p \\ c_{-p}^\dagger
    \end{pmatrix}
    =
    \begin{pmatrix}
        u_p & -v_p^{*} \\
        v_p & u_p^{*}
    \end{pmatrix}
    \begin{pmatrix}
        \tau_p \\ \tau_{-p}^\dagger
    \end{pmatrix},
\end{equation}
where \(u_p = \cos\theta_p\), \(v_p = i\sin\theta_p\), and \(\theta_p\) is the Bogoliubov angle defined by \(\cos\theta_p = \frac{-h + \cos\phi_p - \lambda}{\sqrt{2(\lambda^2 - (\cos\phi_p - h)^2)}}\), with \(\lambda = \sqrt{(h - \cos\phi_p)^2 + ( \sin\phi_p)^2}\). The energy eigenvalues corresponding to each momentum \(p\) are
\begin{equation}
    \epsilon_{\pm p} = \pm \lambda.
    \label{eq:dispersion}
\end{equation}

The ground state texture can then be calculated via the Bogoliubov transformation. By applying the Jordan-Wigner and Fourier transformations, the texture of the ground state becomes
\begin{equation}
    \mathfrak{R}(\ket{\Psi}) = -\sum_{p=1}^{N/2} \ln\left[\text{abs}\left(v_p \cos\frac{\phi_p}{2} - i u_p \sin\frac{\phi_p}{2} \right)^2\right].
\end{equation}

\end{document}